\documentclass[11pt]{article}
\usepackage[margin=1in]{geometry}
\usepackage[utf8]{inputenc}
\usepackage{amsmath,amssymb,amsthm}
\usepackage{bbm}

\usepackage[table]{xcolor}
\usepackage{tcolorbox}
\usepackage{float}
\usepackage{mathtools}
\usepackage{tabularx}

\theoremstyle{definition}
\newtheorem{definition}{Definition}
\newtheorem{example}{Example}
\newtheorem{remark}{Remark}

\theoremstyle{remark}

\theoremstyle{plain}
\newtheorem{theorem}{Theorem}

\newtheorem{prop}{Proposition}
\newtheorem{lemma}{Lemma}
\newtheorem{cor}{Corollary}
\newtheorem{assumption}{Assumption}

\newcommand{\PP}{\mathbbm{P}}

\newcommand{\DD}{\mathcal{D}}

\newcommand{\XX}{\mathcal{X}}
\newcommand{\YY}{\mathcal{Y}}
\newcommand{\HH}{\mathcal{H}}

\renewcommand{\AA}{\mathcal{A}}
\newcommand{\BB}{\mathcal{B}}

\newcommand{\GG}{\mathcal{G}}

\newcommand{\VCdim}{\operatorname{VCdim}}

\newcommand{\rank}{\operatorname{rank}}

\newcommand{\SSD}{\operatorname{SSD}}

\newcommand{\DIS}{\operatorname{DISJ}}
\newcommand{\LCSk}{\operatorname{LCS-{\mathit k}-decision}}

\newcommand{\SM}{\operatorname{SM}}
\newcommand{\ind}{\mathbbm{1}}

\usepackage{hyperref}
\hypersetup{
    colorlinks=true,
    citecolor=blue,
    pdftitle={The Learning and Communication Complexity of Subsequence Containment},
    pdfpagemode=FullScreen,
    }

\title{The Learning and Communication Complexity of Subsequence Containment}
\author{Mason DiCicco\footnote{Computer Science Department, WPI. [mtdicicco@wpi.edu, dreichman@wpi.edu]} \qquad Daniel Reichman$^*$}
\date{\today}

\begin{document}

\maketitle

\begin{abstract}
    We consider the learning and communication complexity of subsequence containment. In the learning problem, we seek to learn a classifier that positively labels a binary string $x$ if it contains a fixed binary string $y$ as a subsequence. In the communication problem, $x$ and $y$ are partitioned between two players, Alice and Bob, who wish to determine if $x$ contains $y$ as a subsequence using a minimal amount of communication. We devise asymptotically tight bounds for the sample complexity (VC dimension) of the learning problem and the communication complexity of the communication problem. Our results illustrate that the sample complexity of our learning problem can be considerably larger when the subsequence occurs in non-contiguous locations. 
\end{abstract}

\section{Introduction}
Given a string $x$ of length $n$ and a string $y$ of length $k \leq n$, we say $y$ is a \emph{subsequence} of $x$ if all of the characters of $y$ appear consecutively (but not necessarily contiguously) within $x$. The subsequence detection problem is to determine, given $x$ and $y$, whether $y$ is a subsequence of $x$. We study the \emph{communication complexity} of this problem: the minimal communication required to compute whether $y$ is a subsequence of $x$ when the characters of $x$ and $y$ are partitioned between two parties, Alice and Bob. We primarily focus on \emph{binary} sequences, but some of our results extend to arbitrary alphabets. 

Subsequence detection has been described as ``one of the most interesting and least studied problems in pattern matching" by \cite{janson2021hidden}. From the learning perspective, non-contiguity seems to arise in certain applications; important features input data could be rather fragmented. In particular, time series (e.g., \cite{kamiyama2017real}), linguistic (e.g., \cite{simard2005translating}) and genetic (e.g., \cite{tahmasebi2020capacity}) features are often non-contiguous. For example, subsequence \emph{anomaly detection} for time series data (as defined by \cite{keogh2007finding}) is a widely studied problem in computer science with a variety of applications. It has been used to detect irregular heartbeats by \cite{hadjem2016st}, machine degradation in manufacturing by \cite{mirylenka2013envelope}, hardware and software faults in data-centers by \cite{pelkonen2015gorilla}, noise within sensors by \cite{bahaadini2018machine}, and spoofed biometric data by \cite{fatemifar2019spoofing}.  Detecting subsequences is also useful in computational biology and has led to deep theoretical questions such as the study of the expected size of the longest common subsequence between two uniformly random strings~\cite{chvatal1975longest}. 

In many applications, the strings that are considered with respect to subsequence detection have lengths in the millions. This can lead to a significant slowdown when attempting to find subsequences in a long data stream. End users with limited computational capacity may share their stream with a party with more computational resources, motivating the question of \emph{communication complexity}. That is, how many bits of communication need to be exchanged in order to allow the party with more computational resources to solve the problem. Additionally, the existence or nonexistence of a subsequence can be useful in \emph{classification}, and it is of practical and theoretical interest to understand how the sample complexity of subsequence-dependent classifiers depends on $n$ and $k$, the lengths of the string and subsequence respectively.

We provide nearly tight bounds for the communication complexity of subsequence detection under a variety of settings (randomized vs. deterministic communication, different partitions of $x,y$ between the two parties, and whether or not the length of $y$ is fixed). We show that, up to log factors, the communication complexity of this problem scales like $O(k)$. This is somewhat surprising as our bounds hold for arbitrary partitions of $x$ and $y$ (not just the natural partition where Alice holds $x$ and Bob holds $y$). We complement these upper bounds by providing a nearly matching lower bound of $\Omega(k)$. We note that the communication complexity of detecting a substring in \emph{contiguous} locations under arbitrary bi-partitions is $\Omega(n)$ for $k=2$, as shown by \cite{golovnev2019string}.


Next, we consider the VC dimension of a family of \emph{subsequence} classifiers defined by containing a fixed subsequence. That is, every classifier is parameterized by $y \in \{0,1\}^k$ such that $x \in \{0,1\}^n$ is classified as $1$ if and only if $x$ contains $y$ as a subsequence (for a precise definition, see Definition~\ref{def:subclassifier}). We prove that the VC dimension of this family of (length $k$) subsequence classifiers is $\Theta(k)$. \cite{sutskever2014sequence} show that in some cases it is beneficial not to assume any upper bound on the length $n$ of the string being classified. Our bounds on the VC dimension easily extend to this case as well. In either case, we are not aware of previous bounds on the VC dimension of this family. Our methods can also be used to bound the VC dimension of classification based on \emph{super}sequences, where a sequence of length $k$ evaluates to $1$ if and only if it occurs in a fixed sequence of length $n \geq k$ as a subsequence. This classification problem resembles trace reconstruction as in \cite{batu2004reconstructing} and may be of independent interest.  

Our methods are straightforward. Lower bounds are proved using reductions from set-disjointness, and upper bounds are proved using simple protocols. Our reduction between subsequence containment and the disjointness problem is useful also in lower bounding the VC-dimension of the subsequence classifiers we study. This expands on the work of \cite{kremer1999randomized} in that it serves as another illustration of the connection between lower bounds from communication complexity and learning theory.

While elementary, our proofs and reductions differ from those  appearing in previous studies of the communication complexity of string related problems (e.g., \cite{sun2007communication, liben2006finding, golovnev2019string, bar2004sketching}). Furthermore, our results uncover a qualitative difference between learning complexity of the contiguous versus the arbitrary. When the classifying subsequence is required to occur in \emph{consecutive} locations within the string, the VC dimension of the classification problem was shown by \cite{golovnev2019string} to be uniformly upper bounded by $O(\log n)$ (regardless of the length of the subsequence $k$). In contrast, when the contiguity requirement is dropped, the VC dimension of subsequence classifiers turns out to be $\Omega(k)$ -- which can be exponentially larger than $\log n$ (for $k=\Omega(n)$). Furthermore, this lower bound holds even if we restrict the occurrence of the classifying substring to have gaps not larger than one; every two characters of the subsequence either appear next to one another or are separated by at most one character. We should remark that we make no attempt to optimize constant factors: we are primarily concerned with the \emph{asymptotic} complexity in the learning setting.

\subsection{Definitions}

\begin{definition}[Alphabets]
Let an \emph{alphabet} $\Sigma$ be any set of symbols (for instance, the binary alphabet $\{0,1\}$). Then, we denote $\Sigma^n$ the set of length-$n$ strings over $\Sigma$. Sometimes we consider strings of \emph{arbitrary} length by $\Sigma^*$.
\end{definition}

\begin{definition}[Subsequence Detection]
\label{def:SSD}
For $n \geq 1$ and alphabet $\Sigma$, define the Boolean function $\SSD^n(x,y)$ whose inputs are strings $x \in \Sigma^n$ and $y \in \Sigma^{*}$ and whose output is $1$ if and only if $y$ is a subsequence of $x$. 

We also define $\SSD^n_k$ (for a positive integer $k \leq n$) where $y$ is guaranteed to belong $\Sigma^k$. 
\end{definition}

\begin{assumption}
When considering strings of arbitrary length, we assume that the length of $y$ does not exceed the length of $x$. Clearly, if the length of $y$ exceeds the length of $x$, then $y$ cannot be a subsequence of $x$. 
\end{assumption}

\begin{assumption}
We will always assume that the alphabet size does not exceed the lengths of the strings (i.e. $|\Sigma| \leq n$) as a string of length $n$ can contain at most $n$ unique symbols. However, we will generally assume a binary alphabet unless otherwise specified. 
\end{assumption}

\begin{example}
\[
\begin{array}{c}
    \SSD^3(010,00) = 1, \\
    \SSD^6(101010, 111) = 1,\\
    \SSD^6_3(120021,211) = 0.
\end{array}
\]
\end{example}

\begin{remark}
A natural question to ask is whether $\SSD^n_k$ belongs to AC$_0$. Namely, whether it can be computed by a Boolean circuit with $\land,\lor$ and $\neg$ gates of polynomial size in $n$ and constant depth. The answer is negative:

\begin{prop}
For all $k \geq 1$, $\SSD^{2k}_{k+1}$ is not in AC$_0$.
\end{prop}
\begin{proof}
Set $y$ to equal $1^{k+1}$ simplifies $\SSD^{2k}_{k+1}$ to the MAJORITY Boolean function which is known not to belong to AC$_0$ (as stated by \cite{jukna2012boolean}). 

\end{proof}
\end{remark}

\subsubsection{Communication Complexity}

We are mainly interested in the \emph{communication} required to compute $\SSD^n$. We now review the relevant definitions from Communication Complexity by \cite{kushilevitz1996communication}.

\begin{definition}[Communication Protocol]
\label{def:protocol}
Let $f : \mathcal{X} \times \mathcal{Y} \to \{0,1\}$. Suppose Alice and Bob are two players holding inputs $x \in \mathcal{X}$ and $y \in \mathcal{Y}$ respectively, with the goal of computing $f(x,y)$. A {\it communication protocol} is a 1-bit message-passing protocol between Alice and Bob. The {\it cost} of a protocol is the maximum number of messages used to compute $f$ over all inputs $(x,y)$.
\end{definition}

\begin{definition}[Communication Complexity]
Let $f : \mathcal{X} \times \mathcal{Y} \to \{0,1\}$. The {\it deterministic communication complexity} of $f$, $D(f)$, is the minimal cost of a deterministic protocol that computes $f$. The {\it randomized communication complexity} of $f$, $R(f)$, is the minimal cost of a randomized protocol that computes $f$ with error probability at most $1/3$.
\end{definition}

Of primary interest in this paper is the function $f$ whose inputs are sets $A,B \subseteq [n]$ and whose value is equal to $1$ if and only if $A$ and $B$ are \emph{disjoint}.

\begin{definition}[Disjointness]
Define $\DIS^n$ as the set-disjointness problem. Given subsets $A,B \subseteq [n]$ respectively, Alice and Bob must determine whether $A$ and $B$ are disjoint. We encode subsets of $[n]$ as their characteristic vectors in $\{0,1\}^n$. Then, $\DIS^n(a,b)$ is defined as the Boolean function whose inputs are characteristic vectors $a,b \in \{0,1\}^n$ and whose output is 1 if and only if the corresponding subsets are disjoint.

As in the definition of $\SSD$, we also consider a restricted variation. Define $\DIS^n_k$ as the problem of set disjointness when $|A| = |B| = k$. Namely, it is the same as $\DIS^n$, with the restriction that both $a$ and $b$ have Hamming weight $k$. We always assume Alice gets $a$ and Bob gets $b$.
\end{definition}

\begin{theorem}[\cite{razborov1990distributional}~\cite{haastad2007randomized}]
\label{th:DISLB}
For all $n \geq k \geq 0$,
\begin{enumerate}
    \item $R(\DIS^n) = \Theta(n)$.
    \item $R(\DIS^n_k) = \Theta(k)$ for every $k \leq n/2$.
    \item $D(\DIS^n_k) = \displaystyle \Theta \left(\log \binom{n}{k} \right)$ for every $k \leq n/2$.
\end{enumerate}
\end{theorem}

Several results in this work rely on \emph{reductions}. A reduction is a high-level way to relate two different problems by transforming one into the other. For our purposes, we define a reduction as follows.

\begin{definition}[Reduction]
Let $f : \mathcal{X} \times \mathcal{Y} \to \{0,1\}$ and $g : \mathcal{A} \times \mathcal{B} \to \{0,1\}$ be two communication problems. We say $f$ \emph{reduces to} $g$ if there exists mappings $\rho: \XX \to \AA$ and $\phi: \YY \to \BB$ such that 
\[
g(\rho(x),\phi(y)) = f(x,y) \text{ for all $(x,y) \in \XX \times \YY$.}
\]
We call the reduction \emph{injective} if $\rho$ and $\phi$ are both injective (i.e. $\rho(x) = \rho(y)$ if and only if $x=y$).
\end{definition}

For instance, we show later that $\DIS$ reduces to $\SSD$. It follows that if we would have a communication protocol $P$ for $\SSD$ then we would have a protocol $Q$ for $\DIS$ with the same cost as of $P$. However, $Q$ cannot contradict existing lower bounds on $\DIS$. This means that $\SSD$ is \emph{at least as ``hard"} as $\DIS$, ignoring technical details discussed in Proposition \ref{prop:redDIS}.

\begin{definition}[Bi-partition]
For any communication problem, the \emph{bi-partition} describes ``who gets what" with respect to the input bits. We mainly focus on the \emph{natural} bi-partition in which Alice holds $x$ and Bob holds $y$. A more general version of this communication problem gives both parties complementary partitions of both $x$ and $y$. For example, Alice may receive the odd-indexed bits of $x$ while Bob receives the even-indexed bits.

We consider both natural and worst-case bi-partitions, and the partition under consideration will always be clear from the context. We sometimes consider a protocol for every possible bi-partition. In this case, the cost of the protocol is the maximal cost over all possible bi-partitions and inputs. 
\end{definition}

\begin{definition}[Communication Matrix]
Let $f: \XX \times \YY \to \{0,1\}$. The {\it communication matrix} $M_f$ is the $|\mathcal{X}| \times |\mathcal{Y}|$ matrix with $(M_f)_{x,y}= f(x,y)$.

For $n \geq k \geq 1$ and alphabet $\Sigma$, we denote the communication matrix $\Sigma^{n \times k} := M_{\SSD^n_k}$. This is a $|\Sigma|^n \times |\Sigma|^k$ binary matrix with 
\[
(\Sigma^{n \times k})_{x,y} = \SSD^n_k(x,y).
\]
\end{definition}

\begin{example}
Let $\Sigma = \{0,1\}$. Then 
\[
\Sigma^{3 \times 2} = 
\left[\begin{array}{c|cccc}
\ & 00 & 01 & 10 & 11  \\
\hline
000 & 1 & 0 & 0 & 0 \\
001 & 1 & 1 & 0 & 0 \\
010 & 1 & 1 & 1 & 0 \\
011 & 0 & 1 & 0 & 1 \\
100 & 1 & 0 & 1 & 0 \\
101 & 0 & 1 & 1 & 1 \\
110 & 0 & 0 & 1 & 1 \\
111 & 0 & 0 & 0 & 1  
\end{array}\right]
\]
\end{example}

\subsubsection{Learning Complexity}

We also consider subsequence containment as a \emph{learning problem}. See \cite{shalev2014understanding} for a more detailed discussion on statistical learning.

\begin{definition}[Subsequence classifiers]
\label{def:subclassifier}
Let $\HH^n_k$ denote the \emph{hypothesis class} of length-$k$ subsequence classifiers acting on strings of length $n$. That is, the collection of functions $\{h_y : y \in \{0,1\}^k\}$ where $h_y : \{0,1\}^n \to \{0,1\}$ such that $h_y(x)=1$ if and only if $y$ is a subsequence of $x$.

Also let $\HH^{n}_*$ denote the collection of subsequence classifiers of \emph{any length} acting on strings of length $n$.
\end{definition}

\begin{definition}[Supersequence classifiers]
\label{def:supclassifier}
Let $\GG^{n}_k$ denote the hypothesis class of length-$n$ supersequence classifiers acting on strings of length $k$. That is, the collection of functions $\{g_x : x \in \{0,1\}^n\}$ where $g_x : \{0,1\}^{k} \to \{0,1\}$ such that $g_x(y)=1$ if and only if $y$ is a subsequence of $x$.

Also let $\GG^n_*$ denote the collection of length-$n$ supersequence classifiers acting on strings of any length.
\end{definition}

\begin{definition}[PAC learning]
Let $\XX = \{0,1\}^n$ be \emph{labelled} by a fixed $h_y \in \HH^n_k$ (which is unknown to the learner). Then, given a distribution $\DD$ over $\XX$, the \emph{loss} of a hypothesis $h_z$ is equal to
\[
L_\DD(h_z) = \PP_{x \sim \DD}(h_y(x) \neq h_z(x)).
\]

A learning algorithm $\AA$ is said to $(\epsilon,\delta)$-\emph{PAC}-learn $\HH^n_k$ if for every distribution $\DD$ over $\mathbf{X}$ and every $h_y \in \HH^n_k$, there exists a \emph{sample size} $N$ such that the following holds\footnote{We focus on the realizable case: For a definition of agnostic PAC-learning please see~\cite{shalev2014understanding}}:

Given $N$ i.i.d samples $\{(x_1,h_y(x_1)),\cdots,(x_N,h_y(x_N))\}$ from $\DD$ as input, $\AA$ outputs (with probability $1-\delta$) a hypothesis $h_z \in \HH^n_k$ with $L_\DD(h_z) < \epsilon$. 
\end{definition}

Finally, we introduce the \emph{Vapnik–Chervonenkis dimension}, a parameter often relevant to statistical learning which is known to be strongly connected to communication complexity (as in \cite{kremer1999randomized}). 

\begin{definition}[VC dimension]
For a finite set $A$, let $\HH$ be any collection of functions $f: A \to \{0,1\}$. Then, a subset $B \subseteq A$ is \emph{shattered} by $\HH$ if for every subset $B' \subseteq B$ there exists a function $f_{B'} \in \HH$ which \emph{realizes} $B'$: For each $b \in B$, $f_{B'}(b)=1$ iff $b \in B'$. The VC dimension of $\HH$, denoted by $\VCdim(\HH)$, is the largest size of a subset of $A$ that is shattered by $\HH$.

For a communication problem $f : \XX \times \YY \to \{0,1\}$, we define $\VCdim(f_\XX)$ as the VC dimension of the hypothesis class parameterized by $\XX$. That is, the collection of functions $\HH_\XX := \{h_x : x \in \XX\}$ such that $h_x(y) = f(x,y)$.
\end{definition}

VC dimension essentially characterizes the number of samples needed to PAC-learn a family of classifiers. 

\begin{theorem}[\cite{hanneke2016optimal,ehrenfeucht1989general}]
\label{th:VCsample}
For any hypothesis class $\HH$, the sample size required to PAC-learn $\HH$ is equal to $\Theta \left( \VCdim (\HH) \right)$. (Here, $\Theta$ hides polynomial dependency of $1/\epsilon$, $1/\delta$, and indicates that the upper bound on the asymptotic learning complexity is tight.)
\end{theorem}

\subsection{Summary of results}

Here we give an informal overview of our results. Firstly, we obtain tight bounds on the deterministic communication complexity of subsequence detection.

\begin{theorem}
\label{th:detCCnat}
For all $n \geq k \geq 0$ and alphabet $\Sigma$, under the natural bi-partition of inputs,
\begin{itemize}
    \item $D(\SSD^n) = n \log |\Sigma|$
    \item $D(\SSD^n_k) = k \log |\Sigma|$
\end{itemize}
That is, the trivial protocol where Bob sends all of $y$ to Alice is optimal.
\end{theorem}

\begin{proof}[Proof Idea]
    
We simply apply the standard \emph{rank argument} to obtain the lower bounds above. The corresponding upper bounds are achieved by the trivial deterministic protocol in both cases.

\end{proof}

\medskip
Next, we obtain similarly tight bounds when the inputs are partitioned arbitrarily between Alice and Bob.

\begin{theorem}
\label{th:detCCarb}
For all $n \geq k \geq 0$ and alphabet $\Sigma$, under an arbitrary bi-partition of inputs,
\begin{itemize}
    \item $D(\SSD^n_k) = \Theta(k \log n)$
\end{itemize}
\end{theorem}

\begin{proof}[Proof Idea]
    
For the lower bound, we make our first \emph{reduction} to disjointness. In particular, we prove that, for any pair of $k$-subsets, $A, B \in \binom{\{1,\cdots,n\}}{k}$, it is possible to construct $x \in \{0,1\}^{3n}$ and $y \in \{0,1\}^{4k}$ such that $y$ is a subsequence of $x$ if and only if $A$ is disjoint from $B$. 

\medskip
In the communication setting, $A$ and $B$ are inputs to the disjointness problem held by Alice and Bob respectively (i.e., the natural bi-partition). However, the characters of $x$ are constructed depending on both $A$ and $B$. As a result, both Alice and Bob both will know some (but not all) of the characters of $x$.

\medskip
Consider any protocol $P$ for the subsequence detection problem under the corresponding bi-partition, and let $A$, $B$ be inputs to the disjointness problem. Alice and Bob may solve the disjointness problem by constructing $x$ and $y$ as above and returning the result of $P(x,y)$. Importantly, $P$ cannot contradict theoretical lower bounds for disjointness, so we obtain a lower bound on the communication complexity of subsequence detection.

\medskip
Due to the non-natural bi-partition, upper bounds are no longer achieved by the trivial deterministic protocol, as this would have cost $O(n)$. Rather, Alice and Bob will iterate over the characters of $y$ and exchange the first \emph{index} where each character occurs, with total cost $O(k \log n)$.

\end{proof}

After this, we give another reduction to disjointness which \emph{preserves} the natural bi-partition, and gives us randomized lower bounds to match Theorem \ref{th:detCCnat}. 

\begin{theorem}
\label{th:rCCnat}
For $n \geq k \geq 0$, under the natural bi-partition of inputs,
\begin{enumerate}
    \item $R(\SSD^n) = \Theta(n)$
    \item $R(\SSD^n_k) = \Theta(k)$
\end{enumerate}
\end{theorem}

\begin{proof}[Proof Idea]
Lower bounds follow in the same manner as above, but with a slightly more involved construction which does preserve the natural bi-partition of inputs. The cost of controlling the bi-partition, unfortunately, is that the length of $y$ is quite large (at least half the length of $x$). However, this issue can be solved with an elementary padding argument which adds an arbitrary number of inconsequential bits to the end of $x$.

\end{proof}

The methods used in the proof of Theorem \ref{th:rCCnat} are also useful in understanding the VC dimension of subsequence classifiers. Importantly, the given reduction from disjointness is \emph{injective}, which enables us to use the following lemma.

\begin{lemma}
\label{lem:VCred}
Let $f : \XX \times \YY \to \{0,1\}$ and $g: \AA \times \BB \to \{0,1\}$ be two communication problems and suppose $f$ has an \textbf{injective} reduction to $g$. Then, 
\begin{enumerate}
    \item $\displaystyle \VCdim(f_\XX) \leq \VCdim(g_\AA)$.
    \item $\displaystyle \VCdim(f_\YY) \leq \VCdim(g_\BB)$.
\end{enumerate}
\end{lemma}

\begin{proof}[Proof Idea]
Let $(\phi,\rho)$ reduce $f$ to $g$ and let $S \subseteq \XX$ be shattered by $f_\XX$. Then, if $\rho$ is injective, it is not hard to see that $\phi(S)$ is also shattered by $g_\AA$, but may have cardinality less than $S$. However, if $\phi$ is also injective, then $|\phi(S)|=|S|$, so shattered sets map to shattered sets of equal size.

\end{proof}

From here it is fairly straightforward to derive tight bounds on the VC dimension of subsequence classifiers, which is our last main result.

\begin{theorem}
\label{th:VCsub}
For all $n \geq \frac65 k \geq 0$,
\begin{enumerate}
    \item $\displaystyle \frac{n}{3} \leq \VCdim(\HH^{n}) \leq n+1$.
    \item $\displaystyle \frac{k}{5} \leq \VCdim(\HH^{n}_{k}) \leq k$.
\end{enumerate}
\end{theorem}

\subsection{Related work}

\subsubsection{Contiguous pattern matching}
In the classical pattern matching problem, we seek to determine whether a string $y$ of length $k$ appears in \emph{contiguous} locations in a string of length $n \geq k$. Let $\SM^n_k$ denote the contiguous string-matching problem. For $k \leq \sqrt{n}$ and arbitrary partitions, \cite{golovnev2019string} prove an upper bound of $D(\SM^n_k) = O(n / k \cdot \log k)$ and a lower bound of $R(\SM^n_k) = \Omega( n / k  \cdot \log \log k)$ bits of communication. We prove significantly smaller bounds for the communication complexity of non-contiguous pattern matching. 

\subsubsection{Non-contiguous pattern matching}
\cite{sun2007communication} and \cite{liben2006finding} proved tight lower bounds for the communication complexity of the $\LCSk$ problem of determining whether two strings of length $n$ have a common subsequence of length $k$ or greater. For example, \cite{sun2007communication} prove that $R(\LCSk) = \Omega(n)$. These works are different than ours as they consider sequences with arbitrary alphabets and we focus on fixed alphabets allowing us to circumvent their strong lower bounds. Additionally, we focus on detecting a subsequence of length $k$ in a string of length $n$ whereas these works focus on computing the largest length of a common subsequence in two strings of length $n$. Consequently, our proof ideas differ from those by~\cite{sun2007communication,liben2006finding}.

Lower bounds on the query complexity of one-sided testers for subsequence-freeness were devised recently by~\cite{ron2021optimal}. While lower bounds for query complexity of testing algorithms were used by \cite{goldreich1998property} to derive lower bounds on VC-dimension, their lower bounds as well as those of \cite{ron2021optimal} do not seem to imply our lower bounds for the VC dimension of classifiers based on the inclusion of a fixed pattern as a subsequence. This is because the lower bound proven by~\cite{ron2021optimal} applies to testers with \emph{one-sided} error and arbitrary alphabets as opposed to our setting where binary sequences are concerned. 

The \emph{deletion channel} takes a binary string as input and independently deletes each bit with fixed probability $d$  (See \cite{janson2021hidden} for a detailed analysis). It was proven by \cite{drmota2012mutual} that the problem of determining the capacity of the deletion channel can be exactly formulated as the subsequence detection problem.

\cite{bringmann2018sketching} show an $\Omega((k/{|\Sigma|})^{|\Sigma|})$ lower bound on the \emph{one-way} communication complexity of subsequence detection. Additionally, they construct a sketch of size $O(k^{|\Sigma|} \log k)$, showing the lower bound is nearly tight.  

\subsubsection{Reconstructing from subsequences}
The problem of reconstructing strings from their subsequences has been previously studied, initiated by \cite{manvel1991reconstruction} and subsequently expanded on by \cite{scott1997reconstructing, dudik2003reconstruction,acharya2015string} which give various conditions on when a string can be reconstructed from its $k$-subsequence decomposition. Our problem differs from the reconstruction problem studied in these works. For example, these works all consider the \emph{multiset}-decomposition of subsequences which includes the multiplicities of each subsequence whereas we only consider the \emph{set}-decomposition for the purposes of subsequence detection.

\subsubsection{VC dimension}
That the \emph{one-way} randomized communication complexity of $f$ is greater than or equal to its VC dimension was observed by \cite{kremer1999randomized}. The crux of our discussion on sample complexity is the exponential gap between the VC dimension of \emph{contiguous} and \emph{non-contiguous} subsequence containment. In particular, \cite{golovnev2019string} showed recently that the contiguous string-matching problem has VC dimension $\Theta(\log(n))$, whereas we prove a linear lower bound on the VC dimension of non-contiguous subsequences.

\section{Communication complexity}

\subsection{Deterministic protocols for the natural bi-partition}

We now consider the natural bi-partition in which Alice holds $x$ and Bob holds $y$. First, we lower bound the deterministic communication complexity of subsequence detection using the well-known \emph{log-rank method}:

\begin{theorem}[\cite{kushilevitz1996communication}]
\label{th:logrank}
For any function $f : \XX \times \YY \to \{0,1\}$,
\[
D(f) \geq \log_2(\rank M_f),
\]
where $\rank M_f$ is equal to the number of linearly independent rows (or columns) of $M_f$.
\end{theorem}

The rank of the communication matrix for subsequence detection, $\Sigma^{n \times k}$, is fairly simple to compute.

\begin{lemma}
\label{prop:detCC}
For all $n \geq k \geq 1$,
\[
\rank(\Sigma^{n \times k}) = |\Sigma|^k
\]
\end{lemma}

\begin{proof}
For simplicity, we assume $\Sigma = \{0,1,\cdots,m\}$ and index the rows and columns of $\Sigma^{n \times k}$ lexicographically (by the sequence $[0^n, 0^{n-1}1, \cdots, m^n]$). Then, it is clear that a string $s \in \Sigma^k$ does not appear as a subsequence in $\Sigma^n$ until the $s$'th string, $0^{n-k}s$, where it appears as a contiguous subsequence. Thus
\begin{itemize}
    \item $i < j \implies (\Sigma^{n \times k})_{i,j} = 0$.
    \item $i = j \implies (\Sigma^{n \times k})_{i,j} = 1$.
\end{itemize}
In particular, the first $|\Sigma|^k$ rows of $\Sigma^{n \times k}$ are a full-rank lower-triangular matrix. 

\end{proof}

\subsubsection{Proof of Theorem \ref{th:detCCnat}}
\begin{proof}

The second lower bound immediately follows from Theorem \ref{th:logrank} applied to Lemma \ref{prop:detCC}:
\[
D(\SSD^n_k) \geq \log \rank(\Sigma^{n \times k}) = k \log |\Sigma|.
\]
The first lower bound follows from the fact that the communication matrix for $\SSD^n$ contains each $\Sigma^{n \times k}$ as a sub-matrix (for every $k \leq n$) and thus has rank equal to $n$.

\medskip
It is easy to achieve these bounds under the natural bi-partition; Bob sends Alice every character of $y$, each requiring $\log |\Sigma|$ bits.  

\end{proof}

\begin{remark} In fact, the same bounds apply to $\SM^n_k$, the contiguous string-matching problem, because $s$ appears contiguously in $0^{n-k}s$. 
\end{remark}

\begin{cor}
For all $n \geq k \geq 1$ and alphabet $\Sigma$, under the natural bi-partition of inputs,
\[
D(\SM^n_k) = k \log |\Sigma|.
\]
\end{cor}

\subsection{Deterministic protocols for arbitrary bi-partitions}

We are also able to tightly bound the deterministic communication complexity of subsequence detection under the \emph{worst-case} bi-partition via a reduction from disjointness.

\begin{prop}
\label{prop:rDISnk}
For all $n \geq k \geq 1$, there exists a bi-partition $B$ such that $\DIS^n_k$ (under the natural bi-partition) reduces to $\SSD^{3n}_{4k}$ under $B$.
\end{prop}
\begin{proof}
Given inputs $a,b \in \{0,1\}^n$ of $\DIS^n_k$ to Alice and Bob, consider the following inputs to $\SSD^{3n}_{4k}$:
\begin{itemize}
    \item $y = 1010 \cdots 10=(10)^{2k}$,
    \item $x = a_1b_1 0 a_2b_2 0 \cdots a_nb_n 0 = (a_ib_i0)_{1 \leq i \leq n}$.
\end{itemize}

This induces the bi-partition of inputs to $\SSD^{3n}_{4k}$ which has Alice hold the $a_i$'s and Bob hold the $b_i$'s. The remaining bits can be partitioned arbitrarily, or even known to both parties simultaneously.

We note that both $a$ and $b$ contain exactly $k$ 1's each. Thus there are $2k$ ``isolated" $1$'s in $x$ (i.e. $y$ is a subsequence) if and only if $a$ and $b$ are disjoint. This completes the proof.

\end{proof}

\subsubsection{Proof of Theorem \ref{th:detCCarb}}

We first state the deterministic lower bound which follows from Proposition \ref{prop:rDISnk}.

\begin{lemma}
\label{prop:reductionLB}
There is a bi-partition of inputs such that
\[
D(\SSD^{n}_{k}) = \Omega \left(\log \binom{n}{k} \right) \text{ for every $k \leq n/2$.}
\]
\end{lemma}
\begin{proof}
A reduction from $\DIS^{n}_{k}$ to $\SSD^{3n}_{4k}$ implies that $\SSD^{3n}_{4k}$ must obey the same communication lower bounds as in Theorem \ref{th:DISLB}.

\end{proof}

With this lower bound in hand, we are ready to prove Theorem \ref{th:detCCarb} by defining a deterministic protocol whose cost is (nearly) asymptotically tight.

\begin{proof}
Alice and Bob first exchange $y$ requiring $O(k \log |\Sigma|)$ bits. Then they compute $i$, the first index in which $x_{i} = y_1$, requiring $O(\log n)$ bits (by exchanging an integer less than or equal to $n$). If there is no such index, then $y$ is not a subsequence of $x$. Otherwise, this reduces to an instance of $\SSD^{n-i}_{k-1}$ with input $x' := x_{i+1}x_{i+2}\cdots x_n$, and $y' = y_2y_3\cdots y_k$. The bi-partition of inputs remains unchanged, although exchanging $y$ is no longer required.

\medskip
Continuing iteratively, we have $D(\SSD^{n}_{k}) = O(k \log |\Sigma| + k \log n) = O(k \log n)$ if we assume $|\Sigma| \leq n$ as in Definition \ref{def:SSD}. This achieves the lower bound in Lemma \ref{prop:reductionLB}, up to a difference of $O(k \log k)$, as
\[
\log \binom{n}{k}  = O \left(\log \left(\frac{n}{k} \right)^k \right) = O\left( k \log n - k \log k \right)
\]
\end{proof}

\subsection{Randomized protocols for the natural bi-partition}

We now give a similar reduction from disjointness which proves a \emph{randomized} lower bound for subsequence detection under the natural bi-partition.

\begin{prop}
\label{prop:redDIS}
$\DIS^n_k$ (injectively) reduces to $\SSD^{3n}_{2n+k}$ under the natural bi-partition.
\end{prop}

\begin{proof}
Let Alice and Bob hold $a,b \in \{0,1\}^n$ respectively. Alice and Bob each construct (without communication) strings $x$ and $y$, which both consist of $n$ ``blocks". Each block will consist of either two or three bits. 

\begin{itemize}
    \item Alice constructs block $i$ equal to $ a_i\overline{a}_i 0$ (i.e. $0 \mapsto 010$ and $1 \mapsto 100$). That is,
    \[
    x =  a_1\overline{a}_1 0 \cdot a_2\overline{a}_2 0  \cdots a_n \overline{a}_n 0
    \]
    
    \item Bob constructs block $i$ equal to $00$ if $b_i=0$, and $010$ if $b_i=1$ (i.e.  $0 \mapsto 00$ and $1 \mapsto 010$). Supposing $b$ contains $k$ ones appearing at indices $i_1 < \cdots < i_k$, then 
    \[
    y = 0^{2(i_1-1)} \cdot 010 \cdot 0^{2(i_2-i_1-1)} \cdot 010 \cdots 0^{2(i_k-i_{k-1}-1)} \cdot 010 \cdot 0^{2(n-i_k-1)}
    \]
    \end{itemize}

If $a$ and $b$ are disjoint, then the $i$'th block of $y$ is a subsequence of the $i$'th block of $x$ for all $i$. Thus, $y$ is a subsequence of $x$. Otherwise, there exists some index $i$ with $a_{i}=b_{i}=1$. Let $\alpha,\beta,\gamma,\delta$ partition $x$ and $y$ around the $i$'th block as follows.
\begin{align*}
    x = \alpha \cdot 1&00 \cdot \gamma \\
    y = \beta \cdot 0&10 \cdot \delta
\end{align*}

Note that both $\alpha$ and $\beta$ contain exactly $2i-2$ zeros, and both terminate in a $0$. Thus, $010 \cdot \delta$ must be a subsequence of $100 \cdot \gamma$. Furthermore, $\gamma$ and $\delta$ contain exactly $n-i$ zeros. If $y$ was a subsequence of $x$, then $010$ must be a subsequence of $100$, which is not the case. Thus, $y$ is not a subsequence of $x$.

As $x$ has length $3n$ and $y$ has length $2n+k$, we have proven that $\DIS^{n}_k(a,b) = \SSD^{3n}_{2n+k}(x,y)$. As every $a$ maps to a unique $x$ (and similarly $b$ to $y$), the reduction is injective, concluding the proof.

\end{proof}

\begin{table}[H]
\caption{
An example instance of the reduction when $a$ and $b$ are not disjoint. Note that there are precisely two zeros in every cell of $x$ and $y$. Thus, for $y$ to be a subsequence of $x$, every zero in $y$ must match a zero in $x$ \emph{in the same column}. However, we cannot match both zeros in the red column because we must also match the bold ``1".}
\begin{center}
\begin{tabularx}{\textwidth} { 
| >{\hsize=.25\hsize \raggedright}X  
| >{\hsize=.25\hsize \centering }X 
| >{\hsize=.25\hsize \centering }X 
| >{\hsize=.25\hsize \centering }X
| >{\hsize=.25\hsize \centering }X
| >{\hsize=.25\hsize \centering }X
| >{\hsize=.25\hsize \centering }X
| >{\hsize=.25\hsize \centering }X
| >{\hsize=.25\hsize \centering }X
| >{\hsize=.25\hsize \centering }X
| >{\hsize=.25\hsize \centering }X
| >{\hsize=.25\hsize \centering \arraybackslash }X | }
\hline
$a$ & 0 & 1 & 1 & 0 & 0 &  1 & 1 & 0 & 1 & 1 & 0\\
\hline 
\rowcolor{gray!20} \cellcolor{white} $x(a)$ &  010 & 100 & 100 & 010 & 010 & \cellcolor{red!20} 100 & \textbf{1}00 & 010 & 100 & 100 & 010 \\
\hline \hline
\rowcolor{gray!20} \cellcolor{white} $y(b)$ &  010 & 00 & 00 & 010 & 00 & \cellcolor{red!20} 0\textbf{1}0 & 00 & 010 & 00 & 00 & 010 \\
\hline 
$b$ & 1 & 0 & 0 & 1 & 0 &  1 & 0 & 1 & 0 & 0 & 1 \\
\hline
\end{tabularx}
\end{center}
\end{table}

\subsubsection{Proof of Theorem \ref{th:rCCnat}}

Note that the reduction above is fairly restrictive; we require the lengths of $x$ and $y$ to be $3n$ and $2n+k$ respectively. However, we may relax this requirement via a simple padding argument.

\begin{proof}
Lower bounds follow from two facts regarding the strings constructed in Proposition \ref{prop:redDIS}:
\begin{enumerate}
     \item As $x$ always has length $3n$, the exact same construction reduces $\DIS^n$ to $\SSD^{3n}$.
    \item Alice may \emph{pad} as many $1$'s as desired to the end of $x$ without compromising the reduction. In particular, if Alice constructs $x \cdot 1^N$ and Bob constructs $y$ as above, we have a reduction from $\DIS^n_k$ to $\SSD^{3n+N}_{2n+k}$ for any value of $N$. Then, for any fixed value of $n$, we may make a simple parameterization, $n' = 3n+N$ and $k' = 2n+k$, to obtain 
    \[
    R(\SSD^{n'}_{k'}) = \Omega(k) = \Omega(k').
    \]
   
\end{enumerate}

Indeed, both of these lower bounds are met (up to constant factors) by the trivial deterministic protocol.

\end{proof}

\begin{remark}
It is interesting to note that, although $y$ appears non-contiguously in $x$, it is highly \emph{constrained} (using the language of \cite{flajolet2006hidden}). That is, $y$ appears in $x$ such that no two consecutive characters are separated by more than $1$ character of $x$.
\end{remark}

\section{VC dimension}

This section amounts to an application of Lemma \ref{lem:VCred}, which connects our communication complexity results to  subsequent sample complexity results.

\subsection{Proof of Lemma \ref{lem:VCred}}

\begin{proof}
(We prove only the first statement as both arguments are identical.) Recall that a reduction from $f$ to $g$ induces two mappings $\rho: \XX \to \AA$ and $\phi: \YY \to \BB$ such that 
\[
g(\rho(x),\phi(y)) = f(x,y) \text{ for all $(x,y) \in \XX \times \YY$.}
\]

Let $S \subseteq \YY$ be shattered by $f_\XX$. By definition, for every $T \subseteq S$, there exists an $x_T \in \XX$ such that (for all $s \in S$),
\begin{equation}
\label{eq:fshat}
f(x_T, s) = 1 \text{ if and only if } s \in T.
\end{equation}

We will show that every such $T$ (which is realized by $x_T$) uniquely maps to a subset $\phi(T)$ which is realized by $\rho(x_T)$. Indeed, as $\rho$ is injective, every $x_T$ maps to a unique $\rho(x_T) \in \AA$ such that
\[
g(\rho(x_T),\phi(s)) = f(x_T,s).
\]
By (\ref{eq:fshat}), and for every $s \in S$, this is equal to $1$ if and only if $s \in T$. Now consider the set $\phi(S) =\{\phi(s) : s \in S\} \subseteq \BB$. As $\phi$ is injective, we have $s \in T$ if and only if $\phi(s) \in \phi(T)$. Thus, $\phi(S)$ is shattered by $g_{\AA}$. 

\end{proof}

\subsection{Disjointness classifiers}

In the same vein as in Proposition \ref{prop:redDIS}, we first calculate the VC dimension of \emph{disjointness classifiers}\footnote{As $\DIS$ is symmetric ($\DIS(a,b) = \DIS(b,a)$), we can refer to the corresponding hypothesis class as $\DIS^n$.} for the sake of lower-bounding the VC dimension of \emph{subsequence classifiers}.

\begin{table}[H]
\caption{For $k=2,3,4,5$ (and various values of $n$) we calculate (by brute force) the largest $S \in \{0,1\}^n$ that is shattered by $\HH^n_k$.}
\begin{center}
\begin{tabular}{ |c|c|c| } 
 \hline
  $k$ & $n$ & $S \subset \{0,1\}^n$ shattered by $\{0,1\}^k$\\ 
  \hline
 2 & 3 & $011,001$ \\ 
 \hline
 3 & 6 & $100001, 111000, 000111$ \\ 
 \hline
 4 & 5 & $10100, 10010, 01010$ \\
 \hline
 5 & 8 & $11000101, 01110010, 10011010, 10110011$ \\
 \hline
\end{tabular}
\end{center}

\end{table}

\begin{lemma}
\label{lem:VCDIS}
For all $n \geq 2k \geq 0$, 
\begin{enumerate}
    \item $\displaystyle \VCdim(\DIS^n) = n$.
    \item $\displaystyle \VCdim(\DIS^n_k) \geq k$.
\end{enumerate}

\end{lemma}
\begin{proof}
The first statement was already proven by \cite{kremer1999randomized}. We give a restate their proof here for completeness.

\medskip
Of course, $\VCdim(\DIS^n)$ cannot be greater than $n$ by a simple surjectivity argument; if $S$ is shattered, then for each $T \in 2^S$, there must exist a unique \emph{classifying} subset $X_T \in 2^{[n]}$. Thus, $|S| \leq n$. 

For the lower bound, we construct a shatterable set of size $n$. In particular, we take the collection of singletons $S=\{\{1\}, \cdots, \{n\}\}$. Indeed, the subset of singletons $\{ \{i_1\}, \cdots, \{i_k\}\} \subseteq S$ is realized by the \emph{complement} of their union, $C = [n] \setminus \{i_1,\cdots,i_k\}$; for each $i \in [n]$, each singleton $\{i\}$ is disjoint from $C$ if and only if $i \in C$.

\medskip
The second statement follows similarly, with a slight loss of generality due to the fact that every classifying subset must have size exactly $k$.

We can again shatter a subset of the singletons, $T=\{\{1\},\cdots,\{k\}\}$ provided that $k \leq n/2$. Every subset $\{ \{i_1\}, \cdots, \{i_\ell\}\} \subset T$ would indeed be realized by the complement $[k] \setminus \{i_1,\cdots,i_\ell\}$ as before, but every classifying subset must have size $k$ (which is not necessarily the case here). Thus, we \emph{pad} these classifiers with elements not in $[k]$. In particular, let $P(m) = \{k+1,\cdots,k+m\}$ consist of $m$ \emph{padding} elements (where $P(0)=\emptyset$). Then, the set $\{ \{i_1\}, \cdots, \{i_\ell\}\}$ is realized by the union
\[
([k] \setminus \{i_1,\cdots,i_\ell\}) \cup P(\ell)
\]
The left-hand side contributes $k-\ell$ elements and the right-hand side $\ell$, for a grand total of $k$ elements as desired. At most, we require $k$ padding elements to realize $T$ itself. Thus, $T$ is shattered when $k \leq n/2$.

\end{proof}
   
\subsection{Proof of Theorem \ref{th:VCsub}}

Theorem \ref{th:VCsub} is now easily proven.

\begin{proof}
Recall the two statements from the proof of Theorem \ref{th:rCCnat}.

\begin{enumerate}
    \item $\DIS^n$ injectively reduces to $\SSD^{3n}$
    \item $\DIS^n_k$ injectively reduces to $\SSD^{3n+N}_{2n+k}$
\end{enumerate}

\medskip
In the first case, the lower bound follows from Lemmas \ref{lem:VCred} and \ref{lem:VCDIS}, and the upper bound via a simple surjectivity argument. 

\medskip
In the second case, again by Lemmas \ref{lem:VCred} and \ref{lem:VCDIS}, we have for $k \leq n/2$,
\[
k \leq \VCdim(\HH^{3n+N}_{2n+k}) \leq 2n+k.
\]

Then, by optimizing over $n$ (i.e. substituting $n=2k$) we obtain $k \leq \VCdim(\HH^{6k+N}_{5k}) \leq 5k$, and dividing $k$ by $5$ completes the proof.

\end{proof}

\begin{example}
What follows is the explicit construction for $n=9$. The shattered strings (which correspond to the singletons $\{1\},\{2\},\{3\}$), are 
\begin{align*}
    s_1 &= 100010010 \\
    s_2 &= 010100010 \\
    s_3 &= 010010100
\end{align*}

\begin{table}[H]
\caption{For completeness, we enumerate every subsequence-classifier $y$ and the corresponding subset $S_y \subseteq \{s_1,s_2,s_3\}$ (i.e. $y$ is a subsequence of \emph{every string} in $S_y$, but \emph{no strings} in the complement).}

\begin{center}
    \begin{tabular}{|c|c|} 
        \hline
        $y$ &  $S_y \subseteq \{s_1,s_2,s_3\}$\\ 
        \hline
        010010010 & $\emptyset$ \\ 
        \hline
        00010010 & $s_1$ \\ 
        \hline
        01000010 & $s_2$ \\
        \hline
        01001000 & $s_3$ \\
        \hline
        0000010 & $s_1, s_2$ \\
        \hline
        0001000 & $s_1, s_3$ \\
        \hline
        0100000 & $s_2, s_3$ \\
        \hline
        000000 & $s_1, s_2, s_3$ \\
        \hline
    \end{tabular}
\end{center}
\end{table}
\end{example}

\begin{remark}
Note that the same bounds apply even for \emph{arbitrary} $n$.  Clearly, $\{0,1\}^n \subset \{0,1\}^*$, so our shattered set is also a subset of $\{0,1\}^*$. Thus, our results apply even to the hypothesis class $\HH^*_k$ whose domain includes sequences of any length.
\end{remark}

Indeed, Theorems \ref{th:VCsub} and \ref{th:VCsample} tells us that non-contiguous subsequences have sample complexity independent of $n$.

\begin{cor}
\label{cor:samplec}
For all $k \geq 0$, $n \geq 6k/5$, the sample complexity of PAC-learning length-$k$ subsequence classifiers is $\Theta(k)$\footnote{As in Theorem \ref{th:VCsample}, $\Theta$ hides polynomial dependency on $1/\epsilon$, $1/\delta$, and indicates that the upper bound on the asymptotic learning complexity is tight.}. 
\end{cor}

\begin{remark}
Efficiently recovering the subsequence based on the training data via the ERM rule seems intractable even in the realizable case: Computing the longest common subsequence of multiple strings was shown by \cite{jiang1995approximation} to be NP-hard. When the LCS is promised to have length $k$, \cite{irving1992two} give an algorithm which computes a common subsequence of length $k$ in $O(N n(n-k)^{N-1}) = O(kn(n-k)^{O(k)})$ time.
\end{remark}

\subsection{Supersequence classifiers}

The previous section analyzed subsequence classifiers, which essentially parameterize the two-argument subsequence detection problem $f(x,y)$ for a fixed subseqeunce $y$ and variable string $x$, that is $f_y(x)$. We may also consider the analogous \emph{supersequence} classifiers, $f_x(y)$, where the string $x$ is fixed and the subsequence $y$ is varied.

\medskip
As the disjointness function $\DIS^n(x,y)$ is symmetric with respect to $x$ and $y$, the following is a direct consequence of Lemmas \ref{lem:VCred} (second statement) and \ref{lem:VCDIS}.

\begin{cor}
For all $n \geq 6k/5 \geq 0$,
\begin{enumerate}
    \item $\displaystyle \frac{n}{3} \leq \VCdim(\GG^n_*) $.
    \item $\displaystyle \frac{k}{5} \leq \VCdim(\GG^n_k)$.
\end{enumerate}
\end{cor}

As usual, we can prove a general upper bound of $n$ by a simple surjectivity argument. However, in some cases it is possible to obtain a tighter bound.

\begin{prop}
\label{prop:GUB}
For all $n \geq 0$ and $k \geq n/2$,
\[
\VCdim(\GG^n_k) \leq n \cdot H(k/n) + 1
\]
where $H(x) = -x\log_2 x - (1-x)\log_2(1-x))$ is the binary entropy function.
\end{prop}

Before proving Proposition \ref{prop:GUB}, we first prove the following Lemma.

\begin{lemma}
Let $E(n,k)$ denote the maximum number of length-$n$ supersequences of any \textbf{fixed} binary string of length $k$. Then,
\[
\VCdim(\GG^n_k) \leq \log_2(E(n,k)) + 1.
\]
\end{lemma}

\begin{proof}
Let $S$ be some shattered set $(y_1,\cdots,y_d)$. By definition, for each $T \subseteq S$, there exists a unique classifying supersequence $x_T$ which realizes $T$. In particular, as the string $x_1$ belongs to precisely $2^{d-1}$ elements of $2^S$ (i.e. unique subsets of $S$), there must exist at least $2^{d-1}$ unique \emph{supersequences} of $y_1$. Thus, if $E(n,k) < 2^{d-1}$, we have a contradiction.

\end{proof}

Thus, proving Proposition \ref{prop:GUB} amounts to calculating the value of $E(n,k)$. To begin with, we may consider the subsequence $y = 1^k$. Then, any $x \in \{0,1\}^n$ is a supersequence of $y$ if and only if $x$ contains $\ell \geq k$ ones. We have by counting that there are exactly $\binom{n}{\ell}$ binary strings which contain precisely $\ell$ zeros. Thus, we obtain the lower bound $E(n,k) \geq \sum_{\ell = k}^n \binom{n}{\ell}$. Interestingly, as noted by \cite{dixon2013longest}, $E(n,k)$ is \emph{invariant} over the choice of subsequence. Thus, this expression is indeed the exact value of $E(n,k)$. For completeness, we give a self-contained proof below.

\begin{lemma}
For all $n \geq k \geq 0$, 
\[
E(n,k) = \displaystyle \sum_{\ell = k}^n \binom{n}{\ell}
\]
\end{lemma}
\begin{proof}

First we define some simplifying notation: For a binary string $z$, we let $z_{-1}$ denote the \emph{tail} of $z$ (i.e. $z_{-1} = z_{2}z_{3}\cdots$). Now, let $n \geq k \geq 0$ and fix a binary string $y \in \{0,1\}^k$. Then, let $E(n,y)$ denote the number of supersequences of $y$. We show inductively that $E(n,y)$ is invariant over the choice of $y$. 

\medskip
Clearly, when $y$ has length $1$, we have $E(n,0) = E(n,1) = 2^{n}-1$, and when $y$ has length $n$ we have $E(n,y)=1$. Then, for $1 < |y| < n$, we have the following. 
\begin{align*}
E(n,y) &= \sum_{x \in \{0,1\}^n} \ind[x \text{ contains $y$ as a subsequence}] \\
&=  \sum_{x \in \{0,1\}^n} ( \ind[x_1=y_1]\ind[x_{-1} \text{ contains $y_{-1}$ as a subsequence}] \\ 
& \qquad \qquad + \ind[x_1 \neq y_1]\ind[x_{-1} \text{ contains $y$ as a subsequence}] ) \\
&= \sum_{z \in \{0,1\}^{n-1}} \ind[z \text{ contains $y_{-1}$ as a subsequence}] + \sum_{z \in \{0,1\}^{n-1}} \ind[z \text{ contains $y$ as a subsequence}] \\
&= E(n-1,y_{-1}) + E(n-1,y). 
\end{align*}

The induction step essentially ``strips" the first bit from $x$, regardless of the value of $y$. Thus,  $E(n,y)$ depends only on the length of $y$. So we may abuse notation and let $E(n,k) := E(n,1^k)$ for the sake of computation. Then, $x \in \{0,1\}^n$ is a supersequence of $1^k$ if and only if $x$ contains $\ell \geq k$ ones. The result follows from the fact that there are exactly $\binom{n}{\ell}$ binary strings which contain precisely $\ell$ ones.

\end{proof}

Now, upper bounding the VC dimension of supersequence classifiers amounts to the following well-known application of Stirling's approximation.

\begin{proof}[Proof of Proposition \ref{prop:GUB}]
The result follows from the estimate (for $\alpha > 1/2$)  
\[
\sum_{k = \alpha n}^n \binom{n}{k} \leq 2^{H(\alpha)n}
\]
as derived in \cite{macwilliams1977theory}.

\end{proof}

\section{Future directions}
There are several questions arising from this work. For the learning problem, we focused on binary alphabets. Studying the VC dimension for larger alphabets is of interest in several applications. Additionally, studying subsequence classifiers based on the occurrence of \emph{multiple} subsequences is of potential interest. Recently the effect contiguity of features in classification tasks on the efficacy of various architectures of neural networks (convolutional vs fully connected nets) was studied by~\cite{shalev2020computational}. Future empirical and theoretical study on how the number and magnitude of gaps influences the success of different learning methods for our subsequence-based classifier could be of interest.

\section{Acknowledgements}
We wish to thank the anonymous reviewers for providing helpful feedback, as well as for bringing \cite{bringmann2018sketching} to our attention. We also thank Cristopher Moore for insightful discussions on this topic.

\bibliographystyle{alpha}
\bibliography{Bib}

\end{document}